\documentclass[twoside,a4paper]{amsart}
\usepackage[utf8]{inputenc}
\usepackage[english]{babel}
\usepackage{amssymb}
\usepackage{dsfont}
\usepackage{amsmath, amsfonts, enumerate}
\usepackage[disable]{todonotes}
\usepackage[numbers]{natbib}
\usepackage{framed}

\usepackage{etoolbox}
\newtoggle{final}
\newtoggle{coimbra2}
\toggletrue{final}
\toggletrue{coimbra2}

\DeclareMathOperator{\divv}{div}

\newcommand{\pair}[1]{\left\langle #1 \right\rangle}

\providecommand{\abs}[1]{\lvert#1\rvert}

\newcommand{\ud}{\mathrm{d}}

\newcommand{\pd}{\partial}

\newcommand{\RR}{{\mathbb R}}

\newcommand{\vol}{{\ud x}}

\newcommand{\Diff}{\mathrm{Diff}}
\newcommand{\Xcal}{\mathfrak{X}}

\newcommand{\SDiff}{\mathrm{SDiff}}
\newcommand{\Xcalvol}{{\Xcal_\vol}}

\newcommand{\LieD}{\mathcal{L}}

\DeclareMathOperator{\ad}{ad}
\DeclareMathOperator{\Ad}{Ad}

\newcommand*\id{\mathrm{id}}

\newcommand*\SO{\mathrm{SO}}

\providecommand{\so}{\mathfrak{so}}

\newcommand*\SE{\mathrm{SE}}
\newcommand*\se{\mathfrak{se}}

\newcommand{\lefttrans}{\mathrm{L}}
\newcommand{\righttrans}{\mathrm{R}}

\usepackage{hyperref}
\usepackage{mathtools}
\mathtoolsset{%
showonlyrefs=true,
}
\MakeRobust{\eqref}

\usepackage{amssymb}
\usepackage{amsthm}

\usepackage{aliascnt}
\usepackage{xparse}
\NewDocumentCommand\xnewtheorem{m o m}
 {%
  \IfNoValueTF{#2}
  {\newtheorem{#1}{#3}[section]}
   {%
    \newaliascnt{#1}{#2}%
    \newtheorem{#1}[#1]{#3}%
    \aliascntresetthe{#1}%
    \expandafter\newcommand\csname #1autorefname\endcsname{#3}%
   }%
 }

\theoremstyle{plain}
\xnewtheorem{theorem}{Theorem}
\xnewtheorem{lemma}[theorem]{Lemma}
\xnewtheorem{corollary}[theorem]{Corollary}
\xnewtheorem{proposition}[theorem]{Proposition}

\theoremstyle{definition}
\xnewtheorem{definition}[theorem]{Definition}
\xnewtheorem{example}[theorem]{Example}
\xnewtheorem{exercise}[theorem]{Exercise}

\xnewtheorem{problem}[Prb]{Problem}

\newtheorem*{remark}{Remark}

\usepackage{color}

 \author{Klas Modin}
 \address{Chalmers University of Technology and University of Gothenburg}
 \email{klas.modin@chalmers.se}
\title[Geometric Hydrodynamics]{Geometric Hydrodynamics: from Euler, to Poincaré, to Arnold}
 \subjclass[2010]{35Q31, 37K65, 70S05}
 \keywords{Euler equations, infinite-dimensional geometry, Lie groups, geometric mechanics}

\thanks{I would like to thank the GMC Network for support, and the organizers in Coimbra for a very friendly and excellent workshop.}

\begin{document}

\begin{abstract}

These are lecture notes for a short winter course at the Department of Mathematics, University of Coimbra, Portugal, December 6--8, 2018.
The course was part of the 13th International Young Researchers Workshop on Geometry, Mechanics and Control.

In three lectures I trace the work of three heroes of mathematics and mechanics: Euler, Poincaré, and Arnold.
This leads up to the aim of the lectures:
to explain Arnold's discovery from 1966 that solutions to Euler's equations for the motion of an incompressible fluid correspond to geodesics on the infinite-dimensional Riemannian manifold of volume preserving diffeomorphisms.
In many ways, this discovery is the foundation for the field of geometric hydrodynamics, which today encompasses much more than just Euler's equations, with deep connections to many other fields such as optimal transport, shape analysis, and information theory.


\end{abstract}

\maketitle



\iftoggle{coimbra2}{
}{

\chapter{The Legacy of Euler, Poincaré, and Arnold}	

\lettrine[lines=3,slope=0pt,findent=0pt,nindent=4pt]{H}{eroes}
are plentiful in mathematics. 
In this chapter we shall meet three of them, and review some of their work, leading up to the geodesic interpretation of incompressible fluid motion.
The material constitutes the background and foundation for \emph{geometric hydrodynamics}---the field of mathematics covered in these lecture notes.\todo{Comment on the relation to `topological hydrodynamics'}\
There are, of course, far more than three contributers to this field; our particular story is abridged to one hero per century.

}

\section{The Incompressible Fluid Equations}\label{sec:incompressible_euler}

\iftoggle{coimbra2}{}{
	\epigraph{Lisez Euler, lisez Euler, c'est notre maître à tous.}{\textit{Pierre-Simon Laplace} \\ on Euler's influence}
}

\noindent It is impossible to overestimate the influence that the Swiss mathematician Leonhard Euler (1707--1783) have had on essentially all of mathematics.
Our story concerns his work on the motion of an incompressible fluid. 

Euler thought of a fluid as a large number of particles moving in a fixed domain.
He argued that, contrary to a solid body, the fluid particles are ``not joined to each other by any bond''\footnote{\citet[Part I, paragraph~3]{Eu1761}, translated to English by Enlin Pan.}.
But he also realized that there are some restrictions on how the particles can move.
In Euler's own words:

\begin{quote}
	At the same time, it cannot be that the motion of all particles of the fluid is bound in \emph{no} way by any law; nor can \emph{any} conceivable motion of a single particle be allowed.
	For since the particles are impenetrable, it is clear that no motion can take place where some particles go through others, or that they penetrate each other. An infinite number of such motions should be excluded, and only the remaining are to be considered, and clearly the task is to determine by which property these remaining possibilities can be distiguished by each other.\footnote{\citet[Part I, paragraph~4]{Eu1761}, translated to English by Enlin Pan.}
\end{quote}

It is stunning how Euler so effortless pin-points the main mechanism behind the complexity of fluid motion---particles moving freely without penetrating each other.
This mechanism is underlying the amazing fluid patterns we see in nature: vortex formations, turbulence, shock-waves, etc.

Euler separates the task he set out in two steps:

\begin{enumerate}
	\item first characterize the set of `possible motions' from the impossible ones;
	\item then, among all possible motions, select the one determined by the principles of mechanics.
\end{enumerate}
In a more modern language the two steps are (1) to determine the constraint manifold, and (2) to formulate Newton's equations while accounting for the constraints through contraint forces.

\begin{remark}
Euler noticed that the set of `possible motions' is infinite-dimensional.
This observation is actually the key to geometric hydrodynamics, although Euler did not know it.
Indeed, in light of Poincaré's work on mechanics on Lie group (\autoref{sec:mechanics_on_lie_groups}), together with modern notions of infinite-dimensional groups, the fact that the set of possible motions can be thought of as an infinite-dimensional Lie group is what unlocked Arnold's great discovery, as we shall see in \autoref{sec:geodesic_interpretation}.\todo{Reformulate this paragraph better}
\end{remark}

By studying infinitesimally small fluid parcels in two ($d=2$) and three ($d=3$) dimensions, Euler derived the `possible motions' as a condition on the vector field $v=(v^1,\ldots,v^d)$ describing the velocity of particles passing through a fluid parcel at position $x=(x^1,\ldots,x^d)$.
By purely geometric considerations (in particular without using the divergence theorem) he arrived at
\begin{equation}
	\sum_{i=1}^d \frac{\partial v^i}{\partial x^i} = 0,
\end{equation}
or, in modern notation,
\begin{equation}
	\divv v = 0.
\end{equation}
The first step, to determine the constraints, was thereby achieved.

Next, addressing dynamics, Euler noticed that the vector field $v$ in general depends also on time (in addition to space).
Thus, the motion of the fluid particles is determined by a time-dependent vector field $v$ in such a way that at every instance in time, it is divergence free.
He presses on by deriving, from Newton's second law, the force $f= (f^1,\ldots,f^d)$ acting on particles in an infinitesimal fluid parcel
\begin{equation}\label{eq:euler_particle_force}
	f^i = \rho \Big( \frac{\partial v^i}{\partial t} + \sum_{j=1}^d v^j \frac{\partial v^i}{\partial x^j} \Big)
\end{equation}
where $\rho =$ const is the mass of the fluid parcel, i.e., the \emph{mass density}.

In the formula \eqref{eq:euler_particle_force} we see one of Euler's many strokes of genius:
the vector field $v$ does not describe the velocity of individual particles, but the (mean) velocity of particles passing through the point $x$ at time $t$.
If $X(t) = (X^1(t),\ldots,X^d(t))$ denotes the position of a specific fluid particle at time $t$, then the particle velocity is given by 
\begin{equation}\label{eq:reconstruction_euler}
	\dot X(t) = v(X(t),t).	
\end{equation}
Furthermore, from Newton's second law we know that the force acting on the particle is given by
\begin{equation}
	f^i = \rho\frac{\ud \dot X^i }{\ud t}  = \rho(\frac{\partial v^i}{\partial t}(X(t),t) + \sum_{j=1}^d\frac{\partial v^i}{\partial x^j} \dot X^j).
\end{equation}
Combining this with \eqref{eq:reconstruction_euler} yields Euler's expression \eqref{eq:euler_particle_force} for the force.
The derived relation between the particle acceleration $\ddot X(t)$ and the time derivative of the fluid vector field $\partial v/\partial t$ turns out to be extremely important in classical field theory; it offers a systematic way of moving between \emph{Eulerian coordinates}, given by the vector field $v$, and \emph{Lagrangian coordinates}, given by the particles' positions and velocities $(X,\dot X)$.

In absence of the external forces and contraints, the equations of motion would now be obtained by equating the force $f$ with the external force $f_e$. 
However, in general, the equations so obtained violates the impenetrability constraint---one needs to add a constraint force $f_c$, so that the vector field $v$ remains divergence free.
Euler realized that this constraint force must have a potential, i.e., it is of the form $f_c = (\partial p/\partial x^1,\ldots,\partial p/\partial x^d)$ for some differentiable function~$p$, which, of course, is the \emph{pressure} of the fluid.\todo{More details needed here.}\
The force balance
\begin{equation}
	f + f_c = f_e
\end{equation}
then yields the final set of equations today known as \emph{Euler's incompressible fluid equations}
\begin{equation}\label{eq:euler_eq_classical}
	\left\{ 
	\begin{aligned}
		&\frac{\partial v^i}{\partial t}  + \sum_{j=1}^d v^j \frac{\partial v^i}{\partial x^j} = -\frac{1}{\rho}\frac{\partial p}{\partial x^i}  + \frac{f_e^i}{\rho}\\
		&\sum_{i=1}^d \frac{\partial v^i}{\partial x^i} = 0 .
	\end{aligned}
	\right.
\end{equation}

What we have seen so far is a classical presentation, not far from the original work of Euler.
In Section~\ref{sec:geodesic_interpretation} below we shall give a completely different, geometric derivation of these equations.
Before that, however, let us give a more modern description in terms of Riemannian geometry.

\subsection*{Formulation on Riemannian Manifolds}

For a more transparent view of the geometric structures underlying Euler's equations \eqref{eq:euler_eq_classical} it is useful to think of the fluid domain as a Riemannian manifold $M$.
This allows us to formulate Euler's equations in a coordinate-free manner.
The remainder of these lecture notes thus requires basic knowledge of Riemannian geometry, for example the topics covered in the excellent book by Lee~\cite{Le1997}.

Recall that the Riemannian structure is given by a smooth field over $M$ which to each $x\in M$ associates an inner product $\pair{\cdot,\cdot}_{\!x}$ on the tangent space $T_xM$.
We shall often use vector analysis dot notation in calculations: $\pair{u,v} \equiv u\cdot v$ and $\abs{u}^2 = \pair{u, u}$.

First recall that the gradient of a scalar differentiable function $f$ is the vector field $\nabla f$ defined so that for all vector fields $v$ on $M$
\begin{equation}
	\pair{\nabla f,v} = \ud f (v).
\end{equation}
The geometric generalization of the gradient, from functions to vector fields, is the \emph{co-variant derivative}.
Indeed, the co-variant derivative of $u$ along $v$, denoted $\nabla_v u$, fulfills three basic properties:
\begin{enumerate}
	\item It is linear in $v$, i.e., $\nabla_{f v+g w} u = f \nabla_v u + g \nabla_w u$ for smooth functions $f$ and $g$;
	\item It is additive in $u$, i.e., $\nabla_v (u+w) = \nabla_v u + \nabla_v w$;
	\item It obeys the product rule in $u$, i.e., $\nabla_v (f u) = f\nabla_v u + \pair{v,\nabla f}u$.
\end{enumerate}
Its local coordinate expression for $v = \sum_i v^i \frac{\pd}{\pd x^i}$ and $u = \sum_i u^i \frac{\pd}{\pd x^i}$ is
\begin{equation}
	\nabla_v u = \sum_{jk} v^j \left( \frac{\pd u^k}{\pd x^j} + \sum_l \Gamma^k_{lj}u^l \right)\frac{\pd}{\pd x^k}
\end{equation}
where $\Gamma^k_{lj}$ are the Christoffel symbols associated with the Riemannian metric.
An important property of the co-variant derivative is
\begin{equation}\label{eq:covar_identity}
	\pair{\nabla_v u,u} = \frac{1}{2}\pair{v,\nabla\abs{u}^2}.
\end{equation}
\begin{exercise}
	Prove the identity \eqref{eq:covar_identity}.
\end{exercise}


Returning now to hydrodynamics, the invariant, geometric formulation of Euler's equations on the Riemannian manifold $M$ is
\begin{equation}\label{eq:euler_incompressible_invariant}
	\left\{
	\begin{aligned}
	&\frac{\partial v}{\partial t} + \nabla_v v = - \nabla p , \\
	&\divv v=0
	\end{aligned}
	\right.
\end{equation}
where, of course, $v$ is the vector field on $M$ describing the fluid in Eulerian coordinates, and $p$ is the pressure function.
(For simplicity, we set the mass density to $\rho \equiv 1$).
We shall discuss the geometrical \emph{origin} of \eqref{eq:euler_incompressible_invariant} in \autoref{sec:geodesic_interpretation} below.
Here, we continue with one of its important conservation laws: conservation of energy.

The energy functional is the sum of the kinetic energies for all the fluid particles and is thus given by
\begin{equation}
	E(v) = \frac{1}{2}\int_M \abs{v}^2 \vol ,
\end{equation}
where $\vol$ denotes the standard volume element associated with the Riemannian structure on $M$.
To prove that $E$ is conserved we shall need a result that goes back to Helmholtz~\cite{He1858} work on fluid dynamics in the 1850s.

\begin{lemma}[Helmholtz decomposition]\label{lem:helmholtz}
	Let $M$ be a compact manifold, possibly with boundary, and let $u$ be a $C^1$ vector field on $M$.
	Then there exist a $C^1$ vector field $v$ and a $C^2$ function $f$ such that
	\begin{equation}
		u = v + \nabla f \quad\text{with}\quad \divv v = 0 \quad\text{and}\quad v|_{\partial M} = 0.
	\end{equation}
	Furthermore, the components $v$ and $\nabla f$ are orthogonal in the $L^2$ sense
	\begin{equation}
		\pair{\nabla f,v}_{L^2} \coloneqq \int_M \nabla f\cdot v \, \vol = 0.
	\end{equation}

\end{lemma}

\begin{proof}
	For all technical details of the proof (especially elliptic PDE theory), we refer to \cite{Ta1996a}.
	A brief sketch of the proof goes as follows.
	Given $u$ which is $C^1$, consider the Poisson equation
	\begin{equation}
		\Delta f = -\divv u
	\end{equation}
	with the inhomogeneous Neumann boundary conditions
	\begin{equation}
		\nabla f \cdot \textbf{n}\Big|_{\partial M} = u\Big|_{\partial M}.
	\end{equation}
	From elliptic PDE theory we know that there exists a $C^2$ solution. 
	Now set $$v\coloneqq u -\nabla f.$$
	By construction we have $\divv v = 0$ and $\textbf{n}\cdot v|_{\partial M} = 0$.
	Orthogonality between the terms follows from Stoke's theorem, since
	\begin{equation}
		\int_M \nabla f \cdot v \,\vol = - \int_M f \underbrace{\divv v}_{=0} \,\vol + \int_{\partial M} f \underbrace{\textbf{n}\cdot v}_{=0}\,\ud S.
	\end{equation}
\end{proof}

To prove conservation of energy,
let $v= v(x,t)$ and $p=p(x,t)$ be a solution to \eqref{eq:euler_incompressible_invariant}.
Then
\begin{align}
	\frac{\ud }{\ud t}E(v) &= -\int_M v\cdot (\nabla_v v+\nabla p) \,\vol \\
	&= -\int_M v\cdot \nabla_v v \,\vol - \int_M v\cdot \nabla p \, \vol .
\end{align}
That the second term vanishes follows directly from Helmholtz's \autoref{lem:helmholtz}.
Using the property \eqref{eq:covar_identity} of the co-variant derivative we then have
\begin{align}
	\frac{\ud }{\ud t}E(v) &= -\int_M v\cdot \nabla_v v \,\vol = -\int_M v\cdot \frac{1}{2}\nabla \abs{v}^2 \,\vol = 0
\end{align}
where the last identity also follows from \autoref{lem:helmholtz} since $v$ is divergence free.

\section{Mechanics on Lie groups}\label{sec:mechanics_on_lie_groups}
In a two-page paper from 1901 Henri Poincaré~\cite{Po1901} -- our second hero of the \iftoggle{coimbra2}{lecture notes}{chapter} -- derived the differential equations for mechanical systems evolving on general (finite dimensional) Lie groups.
He arrived at this through the dynamics of rotating rigid bodies in liquids, but he gave no clear motivation for why he studied these equations in such great generality (for any \emph{continuous transformation group}, what we today call a Lie group).
He did not return to this work after it was published.
Nevertheless, he understood the importance of these equations.
As we shall see in \autoref{sec:geodesic_interpretation} below, the structure studied in Poincaré's short paper constitutes the spine of geometric hydrodynamics, although it took over 60 years, and another brilliant mind, to realize this.

Let us begin with a brief review of Lie groups.
They are groups that are also manifolds, where the group multiplication and inversion are smooth operations.
The unit element in a Lie group $G$ is denoted $e$.
For a fixed $h\in G$, the corresponding left and right translation operators are given by
\begin{equation}
	\lefttrans_h\colon G\to G, \; g\mapsto hg \qquad\text{and}\qquad \righttrans_h\colon G\to G, \; g\mapsto gh .
\end{equation}
The \emph{lifted left and right actions} of $G$ on its tangent bundle $TG$ are given by
\begin{equation}
	h\cdot (g,\dot g) = T\lefttrans_h(g,\dot g) \qquad\text{and}\qquad (g,\dot g)\cdot h = T\righttrans_h(g,\dot g).
\end{equation}
A Lie group $G$ also acts on its Lie algebra $\mathfrak g = T_e G$ by the adjoint operator
\begin{equation}
	\Ad_g\colon\mathfrak{g}\to\mathfrak{g},\quad \xi \mapsto T_e (\lefttrans_g \circ \righttrans_{g^{-1}})\xi.
\end{equation}
Notice that $\Ad_g$ is the derivative at $e$ of the \emph{inner automorphism} $h\mapsto ghg^{-1}$.
If $\mathrm I(g,h) = ghg^{-1}$ then the Lie bracket $[\cdot,\cdot]\colon \mathfrak g\times\mathfrak g\to \mathfrak g$ is the derivative of $I$ at $(e,e)$
\begin{equation}
	[\eta,\xi] = T_{(e,e)}\mathrm I(\eta,\xi).
\end{equation}
If $g(t)$ and $h(s)$ are paths in $G$ such that $g'(0) =\eta$ and $h'(0) = \xi$, then
\begin{align}
	[\eta,\xi] &= \frac{\ud}{\ud t} \frac{\ud}{\ud s}\Bigg|_{t=s=0} g(t)h(s)g(t)^{-1} \\
	&= \frac{\ud}{\ud t}\Bigg|_{t=0} \Ad_{g(t)}\xi \\
	&= \frac{\ud}{\ud t}\Bigg|_{t=0} \left( T_e \lefttrans_{g(t)}\xi - T_e\righttrans_{g(t)^{-1}}\xi \right).
\end{align}
These formul\ae\ are useful for computing the Lie bracket, especially for infinite dimensional Lie groups as we shall see in \autoref{sec:geodesic_interpretation} below.
If $G$ is a matrix Lie group, then we have the following formul\ae\
\begin{equation}
	h\cdot(g,\dot g) = (hg,h\dot g),\qquad \Ad_g(\xi) = g\xi g^{-1},\qquad [\xi,\eta] = \xi\eta-\eta\xi .
\end{equation}

Back now to Poincaré.
He started from a finite dimensional Lie group $G$ and considered a Lagrangian on $TG$ fulfilling, for all $(g,\dot g)\in TG$ and all $h\in G$, the left invariance property
\begin{equation}\label{eq:Linvariance}
	L\big(h\cdot(g,\dot g)\big) = L(g,\dot g).
\end{equation}

\begin{lemma}
	Let $\mathfrak g$ denote the Lie algebra of $G$.
	Then any Lagrangian fulfilling \eqref{eq:Linvariance} is of the form
	\begin{equation}
		L(g,\dot g) = \ell(g^{-1}\cdot\dot g)
	\end{equation}
	for the function $\ell\colon \mathfrak g\to \RR$ defined by $\ell(\xi) = L(e,\xi)$.
\end{lemma}

\begin{proof}
	From the invariance, taking $h=g^{-1}$, we have
	\begin{equation}
		L(g,\dot g) = L(g^{-1}\cdot (g,\dot g)) = L(e,g^{-1}\cdot \dot g) = \ell(g^{-1}\cdot \dot g).
	\end{equation}
\end{proof}

As for any Lagrangian system, the Euler--Lagrange equations for $L$ yields a flow on $TG$.
However, since $G$ is a Lie group, the tangent bundle is trivializable as $TG\simeq G\times \mathfrak g$ by the mapping
\begin{equation}\label{eq:left_trivialization}
	(g,\dot g)\mapsto (g,\underbrace{T_{g}\lefttrans_{g^{-1}}\dot g}_{\xi}).
\end{equation}
What Poincaré did was to derive the Euler--Lagrange equations expressed in the trivialized coordinates $(g,\xi)$.
Due to the invariance~\eqref{eq:Linvariance} of the Lagrangian, the resulting equation for $\xi$ is independent of $g$.

\begin{theorem}[Poincaré~\cite{Po1901}]\label{thm:euler_poincare_eq}
	Let $g\colon [0,1]\to G$ be a solution to the Euler--Lagrange equations for $L$, thus extremizing the action functional
	\begin{equation}
		A(g) = \int_0^1 L(g(t),\dot g(t))\ud t .
	\end{equation}
	Then the path $\xi\colon [0,1]\to \mathfrak g$ defined by \eqref{eq:left_trivialization} fulfills the \emph{Euler--Poincaré equation}
	\begin{equation}\label{eq:euler_poincare}
		\frac{\ud}{\ud t}D \ell(\xi) - \ad_\xi^* D\ell(\xi) = 0,
	\end{equation}
	where $\ad^*_\xi$ is the co-adjoint infinitesimal action of $\xi$ on $\mathfrak g^*$, defined by
	\begin{equation}\label{eq:adstar_def}
		\pair{\ad^*_\xi \mu, \zeta} = \pair{\mu, [\xi,\zeta]},\qquad \forall\,\zeta\in \mathfrak g.
	\end{equation}
\end{theorem}

Before we prove the theorem, let us just point out that the path $g(t)$ can be reconstructed from $\xi(t)$ and the initial point $g(0) = g_0$ by the \emph{reconstruction equation}
\begin{equation}\label{eq:reconstruction}
	\dot g(t) = T_e \lefttrans_{g(t)} \xi(t).
\end{equation}
Thus, the Euler--Poincaré equation~\eqref{eq:euler_poincare} together with the reconstruction equation~\eqref{eq:reconstruction} correspond to the full Euler--Lagrange equations for $L$.

For the proof of \autoref{thm:euler_poincare_eq} we shall need the following result.

\begin{lemma}\label{lem:variation_poincare}
	Let $g\in C^2([0,1]\times (-\delta,\delta),G)$ for some $\delta > 0$. 
	If
	\begin{equation}
		\xi(t,\epsilon) \coloneqq T_{g(t,\epsilon)}\lefttrans_{g(t,\epsilon)^{-1}}\frac{\partial g(t,\epsilon)}{\partial t}
	\end{equation}
	and
	\begin{equation}
		\eta(t,\epsilon) \coloneqq T_{g(t,\epsilon)}\lefttrans_{g(t,\epsilon)^{-1}}\frac{\partial g(t,\epsilon)}{\partial \epsilon}
	\end{equation}
	then
	\begin{equation}
		\frac{\partial \xi(t,\epsilon)}{\partial \epsilon} - \frac{\partial \eta(t,\epsilon)}{\partial t} = [\xi(t,\epsilon),\eta(t,\epsilon)].
	\end{equation}
\end{lemma}

\begin{proof}
	For simplicity, let us assume that $G$ is a matrix Lie group.
	Then
	\begin{equation}
		\xi = g^{-1}\frac{\pd g}{\pd t}\qquad\text{and}\qquad \eta = g^{-1}\frac{\pd g}{\pd \epsilon}
	\end{equation}
	Thus,
	\begin{equation}
		\frac{\pd\xi}{\pd \epsilon} = \frac{\pd g^{-1}}{\pd \epsilon}\frac{\pd g}{\pd t} + g^{-1}\frac{\pd^2 g}{\pd t\pd\epsilon}
	\end{equation}
	and
	\begin{equation}
		\frac{\pd\eta}{\pd t} = \frac{\pd g^{-1}}{\pd t}\frac{\pd g}{\pd \epsilon} + g^{-1}\frac{\pd^2 g}{\pd t\pd\epsilon}
	\end{equation}
	In the difference between the two, the second order derivatives cancel, and we get
	\begin{equation}
		\frac{\pd\xi}{\pd \epsilon} - \frac{\pd\eta}{\pd t} = \frac{\pd g^{-1}}{\pd \epsilon}\frac{\pd g}{\pd t} - \frac{\pd g^{-1}}{\pd t}\frac{\pd g}{\pd \epsilon}.
	\end{equation}
	For the derivative of $g^{-1}$ with respect to $t$ (and similarly for $\epsilon$), we have
	\begin{equation}
		0 = \frac{\pd}{\pd t} g^{-1}g = \frac{\pd g^{-1}}{\pd t}g + g^{-1}\frac{\pd g}{\pd t}
		\iff \frac{\pd g^{-1}}{\pd t} = - g^{-1}\frac{\pd g}{\pd t} g^{-1}.
	\end{equation}
	Thus
	\begin{equation}
		\frac{\pd\xi}{\pd \epsilon} - \frac{\pd\eta}{\pd t} = - \underbrace{g^{-1}\frac{\pd g}{\pd \epsilon}}_{\eta} \underbrace{g^{-1}\frac{\pd g}{\pd t}}_{\xi} + \underbrace{g^{-1}\frac{\pd g}{\pd t}}_{\xi} \underbrace{g^{-1}\frac{\pd g}{\pd \epsilon}}_{\eta} = [\xi,\eta].
	\end{equation}
	Thereby the stated result is obtained.
\end{proof}

We are now in position to prove \autoref{thm:euler_poincare_eq}.

\begin{proof}[Proof of \autoref{thm:euler_poincare_eq}]
	Let $g_\epsilon(t)$ be a variation of $g(t)$, i.e., $g_0(t) = g(t)$.
	Since $g$ extremizes $A$ we have
	\begin{equation}
		0 = \frac{\ud}{\ud \epsilon}\Big|_{\epsilon=0} A(g_\epsilon) = \frac{\ud}{\ud \epsilon}\Big|_{\epsilon=0}\int_0^1 \ell(\underbrace{g_\epsilon^{-1}\cdot \dot g_\epsilon(t)}_{\xi_\epsilon(t)})\ud t
		= \int_0^1 \pair{D\ell(\xi(t)),\frac{\ud}{\ud \epsilon}\Big|_{\epsilon=0}\xi_\epsilon(t)} \,\ud t.
	\end{equation}
	From \autoref{lem:variation_poincare} it then follows that
	\begin{equation}
		0 = \int_0^1 \pair{D\ell\big(\xi(t)\big), \dot\eta(t) + [\xi(t), \eta(t)]  } \,\ud t,
	\end{equation}
	where
	\begin{equation}
		\eta(t) \coloneqq g(t)^{-1}\cdot  \frac{\ud}{\ud \epsilon}\Big|_{\epsilon=0} g_\epsilon(t)
	\end{equation}
	Since the variation $g_\epsilon$ is fixed at the end-points, it follows that $\eta(0) = \eta(1) = 0$.
	Thus, integration by parts with respect to $t$ gives
	\begin{equation}
		0 = \int_0^1 \pair{-\frac{\ud}{\ud t}D\ell\big(\xi(t)\big) + \ad^*_{\xi(t)}D\ell\big(\xi(t)\big),\eta(t) } \ud t
	\end{equation}
	Since the mapping $T_g G \ni V \mapsto g^{-1}\cdot V \in \mathfrak{g}$ is surjective, the statement in the theorem now follows from the fundamental lemma of calculus of variations.
\end{proof}

\begin{example}[Free rigid body]
	The rigid body equations, also derived by Euler, describes the motion of a rigid body floating in space without influence of any external forces.
	These equations constitute an Euler--Poincaré system as follows.
	Let the `reference configuration' of the rigid body be described by a compact domain $\Omega\subset \RR^3$ such that its center of mass (for simplicity we assume that the mass density is unitary) is located at the origin.
	All possible configurations are given by
	\begin{equation}
		\Omega_{(R,\mathrm r)} \coloneqq R\Omega + \mathrm r \subset \RR^3
	\end{equation}
	where $R$ is a rotation matrix (thus belonging to the Lie group $\SO(3)$) and $\mathrm r$ is a translation vector.
	Thus, the motion is described by a path $t\mapsto (R(t),\mathrm r(t))$.

	Notice that
	\begin{equation}
		\Omega = \Omega_{I,\mathrm 0}
	\end{equation}
	and if $B\in\SO(3)$ and $b\in \RR^3$ then
	\begin{equation}
		B\Omega_{(R,\mathrm r)} + \mathrm b  = \Omega_{(BR,B\mathrm r + \mathrm b)} .
	\end{equation}
	This means that the configuration space has a group structure. 
	It is the special Euclidean group $\SE(3) = \SO(3)\ltimes \RR^3$: the \emph{semi-direct product} constructed from the left action of $\SO(3)$ on $\RR^3$.
	
	The velocity of an infinitesimal mass element initially located at $X\in\Omega$ is
	\begin{equation}
		\dot X(t) = \dot R(t) X + \dot{\mathrm r}(t).
	\end{equation}
	This is the Lagrangian velocity, corresponding directly to the left hand side of \eqref{eq:reconstruction_euler} in the reconstruction equation for the Euler fluid equations.\footnote{One can view a rigid body as a `fluid', only with a stronger constraint on the fluid particles: instead of preserving volume it should now preserve the metric tensor.}\
	The Lagrangian on $T\SE(3)$ is given by the kinetic energy.
	Assuming constant mass density, it is
	\begin{align}\label{eq:rb_lagrangian}
		L(R,\mathrm r, \dot R, \dot{\mathrm r}) &= \frac{1}{2}\int_\Omega \abs{\dot X}^2\, \ud X = 
		\frac{1}{2}\int_\Omega \abs{\dot R X}^2\ud X +
		\frac{1}{2}\int_\Omega \abs{\dot{\mathrm r}}^2\ud X + 
		\int_\Omega (\dot R X)\cdot \dot{\mathrm r}\ud X \\
		&= \frac{1}{2}\int_\Omega \abs{\dot R X}^2\ud X +
		\frac{1}{2} \abs{\dot{\mathrm r}}^2 \underbrace{\int_\Omega \ud X}_{m} + 
		\Big(\dot R \underbrace{\int_\Omega X\ud X}_{0} \Big)\cdot \dot{\mathrm r} \\
		&= \frac{1}{2}\int_\Omega \abs{\dot R X}^2\ud X +
		\frac{1}{2} m\abs{\dot{\mathrm r}}^2 .
	\end{align}
	The lifted left action of $(B,\mathrm b)\in \SE(3)$ on $(R,\mathrm r, \dot R, \dot{\mathrm r}) \in T\SE(3)$ is
	\begin{equation}
		(B,\mathrm b)\cdot(R,\mathrm r, \dot R, \dot{\mathrm r}) = (BR,B\mathrm r + \mathrm b, B\dot R, B\dot{\mathrm r}).
	\end{equation}
	The Lagrangian \eqref{eq:rb_lagrangian} is invariant with respect to this action:
	\begin{equation}
		L\big((B,\mathrm b)\cdot(R,\mathrm r, \dot R, \dot{\mathrm r}) \big) = 
		\frac{1}{2}\int_\Omega \abs{B\dot R X}^2\ud X +
		\frac{1}{2} m\abs{B\dot{\mathrm r}}^2 = L(R,\mathrm r, \dot R, \dot{\mathrm r})
	\end{equation}
	where the last equality follows since $B$ is a rotation matrix, thus preserving lengths.
	In summary, we see that the free rigid body is a left invariant Lagrangian system on $T\SE(3)$, so we are in position to apply \autoref{thm:euler_poincare_eq}.

	First, the Lie algebra of $\SE(3)$ is given by $\se(3) = \so(3)\ltimes \RR^3$ with Lie bracket
	\begin{equation}
		[(\hat\omega,\mathrm{v}),(\hat\sigma,\mathrm{w})] = (\hat\omega\hat\sigma-\hat\sigma\hat\omega,\omega\mathrm{w}-\sigma\mathrm{v}),
	\end{equation}
	where $\hat\omega_i$ are skew symmetric matrices.
	One usually identifies $\so(3)$ with $\RR^3$ via
	\begin{equation}
		\hat\omega \leftrightarrow \omega ,\qquad 
		\begin{pmatrix} 
			0 & -\omega^3 & \omega^2 \\
			\omega^3 & 0 & -\omega^1 \\
			-\omega^2 & \omega^1 & 0
		\end{pmatrix}
		\leftrightarrow
		\begin{pmatrix}
			\omega^1 \\ \omega^2 \\ \omega^3
		\end{pmatrix}.
	\end{equation}
	With this identification, the Lie bracket is given in terms of the vector cross product
	\begin{equation}
		[(\omega,\mathrm{v}),(\sigma,\mathrm{w})] = (\omega\times\sigma,\omega\times \mathrm{w}-\sigma\times\mathrm{v}).
	\end{equation}

	Let us now give the reduced Lagrangian $\ell$ on $\se(3)$.
	Since the first term in the Lagrangian \eqref{eq:rb_lagrangian} is quadratic in $\dot R$, we have
	\begin{equation}
		\ell(\omega,\mathrm{v}) = \frac{1}{2}\omega\cdot\mathbb{I}\omega + \frac{m}{2}\abs{\mathrm{v}}^2,
	\end{equation}
	where $\mathbb I$ is a symmetric matrix: the \emph{moments of inertia tensor} of the rigid body.
	Identifying the dual $\se(3)^*$ with $\se(3)$ via the pairing
	\begin{equation}
		\pair{(\pi,p),(\omega,\mathrm v)} = \pi\cdot\omega + p\cdot v
	\end{equation}
	we get
	\begin{equation}
		D\ell(\omega,\mathrm v) = \Big(\mathbb{I}\omega, m\mathrm v\Big).
	\end{equation}
	Furthermore, for the definition \eqref{eq:adstar_def} of the co-adjoint action $\ad^*$ we have
	\begin{align}
		\pair{(\pi,p),[(\omega,\mathrm{v}),(\sigma,\mathrm{w})]} 
		&= \pi\cdot (\omega\times\sigma) + p\cdot \left( \omega\times \mathrm{w}-\sigma\times\mathrm{v}\right) \\
		&= \sigma\cdot (\pi\times\omega) + \mathrm{w}\cdot ( p\times \omega ) - \sigma\cdot(\mathrm{v}\times p ) \\
		&= \pair{\underbrace{(\pi\times\omega+p\times \mathrm{v}, p\times \omega)}_{\ad^*_{(\omega,\mathrm{v})}(\pi,p)},(\sigma,\mathrm{w} )}.
	\end{align}
	In particular,
	\begin{equation}
		\ad^*_{(\omega,v)}D\ell(\omega,v) = \left( \mathbb{I}\omega\times \omega + m\underbrace{\mathrm v \times \mathrm v}_0 ,m\mathrm v\times \omega \right)
	\end{equation}
	By \autoref{thm:euler_poincare_eq} we now get the equations of motion as an Euler--Poincaré equation for the variables $\omega$ and $v$ as
	\begin{equation}\label{eq:rb_equations_lagrangian}
	\left\{
	\begin{aligned}
		&\mathbb{I}\dot\omega - \mathbb{I}\omega\times\omega = 0 \\ 
		&\dot{\mathrm v} - \mathrm v\times\omega = 0.
	\end{aligned}
	\right.
	\end{equation}
	The variable $\omega$ is the \emph{body angular velocity}, i.e., the angular velocity expressed in the coordinate frame of the moving body.
	The variable $v$ is the \emph{body center of mass velocity}, i.e., the velocity vector for the center of mass expressed in the coordinate frame of the moving body.
	Notice that the equation for $\omega$ is independent of $v$ reflecting the well-known fact that the angular acceleration of a free rigid body is independent of the velocity of its center of mass.
	The reconstruction equation~\eqref{eq:reconstruction}, recovering $(R(t),\mathrm r(t))$ from a solution $(\omega(t),\mathrm v(t))$ of \eqref{eq:rb_equations_lagrangian}, becomes
	\begin{align}
		\dot R(t) &= R(t)\omega(t) \\
		\dot{\mathrm r}(t) &= R(t)\mathrm v(t) .
	\end{align}
\end{example}

\section{Geodesic Interpretation of Fluid Motion}
\label{sec:geodesic_interpretation}

Vladimir Arnold, who had carefully studied the work of both Euler and Poincaré, realized that Poincaré's framework for mechanics on Lie groups also makes sense (at least formally) in infinite dimensions: he realized that the group of diffeomorphisms on a manifold $M$, here denoted $\Diff(M)$, can be viewed as an infinite dimensional Lie group with composition and inversion as group operations.
If we take a path of diffeomorphisms $t\mapsto \varphi(t,\cdot) \in \Diff(M)$ and differentiate it with respect to $t$ we obtain, for each $x\in M$, an element in $T_{\varphi(x)}M$.
Thus, the tangent space $T_\varphi\Diff(M)$ at $\varphi \in \Diff(M)$ consists of functions $\dot\varphi\colon M\to TM$ such that\todo{Add illustration here.}
\begin{equation}\label{eq:diff_tangent_space}
	\dot\varphi(x) \in T_{\varphi(x)}M.
\end{equation}
The Lie algebra of $\Diff(M)$ is, as usual, the tangent space at the identity diffeomorphism $\id(x) = x$ (which is the unit element of the group).
Thus, from \eqref{eq:diff_tangent_space} it follows that $T_\id\Diff(M)$ are (smooth) sections of the tangent bundle $TM$, i.e., vector fields.
We often use the notation $\Xcal(M)= T_\id\Diff(M)$. 
Notice that if $M$ has a boundary, it is required that vector fields in $T_\id\Diff(M)$ are tangential to it (otherwise it does not generate a diffeomorphism).

As in the finite dimensional case discussed in \autoref{sec:mechanics_on_lie_groups}, let us now derive the lifted action of $\eta \in \Diff(M)$ on $(\varphi,\dot\varphi)\in T\Diff(M)$.
In the finite dimensional case, we were mainly interested in left actions.
Here we focus instead of right actions.
One reason, as we shall see, is that the Lagrangian on $T\Diff(M)$ whose corresponding Euler--Poincaré equation is the incompressible Euler equation~\eqref{eq:euler_incompressible_invariant} is right invariant.
Another reason, which is perhaps more important, is that when one works with Sobolev completions of $\Diff(M)$, only right multiplication is smooth (left multiplication is continuous but not Lipschitz).

Let $\varphi = \varphi(t)$ be a path on $\Diff(M)$ and let $\eta\in\Diff(M)$. 
Then, by definition, the right lifted action of $\Diff(M)$ on $(\varphi,\dot\varphi)\in T\Diff(M)$ is
\begin{equation}
	\frac{\ud}{\ud t}\varphi\circ\eta = \dot\varphi\circ\eta .
\end{equation}
In other words, the tangent derivative of the right multiplication operator $\righttrans_\eta\colon \Diff(M)\to\Diff(M)$ is given by
\begin{equation}
	T\righttrans_\eta(\varphi,\dot\varphi) = (\varphi\circ\eta,\dot\varphi\circ\eta).
\end{equation}
Similarly, for the left lifted action we have by the chain rule
\begin{equation}
	\frac{\ud}{\ud t}\eta\circ\varphi = D\eta\circ\varphi\, \dot\varphi 
\end{equation}
which implies
\begin{equation}
	T\lefttrans_\eta(\varphi,\dot\varphi) = (\eta\circ\varphi,D\eta\circ\varphi\,\dot\varphi).
\end{equation}
This allows us to defined the adjoint action of $\eta\in\Diff(M)$ on the algebra element $v\in\Xcal(M)$
\begin{equation}
	\Ad_\eta v = T_{\id}(\lefttrans_\eta\circ \righttrans_{\eta^{-1}})v = (T\lefttrans_{\eta}\circ T\righttrans_{\eta^{-1}})(\id,v)
	= D\eta (v\circ\eta^{-1}).
\end{equation}
To obtain the infinitesimal adjoint action, i.e., the Lie bracket on $\Xcal(M)$, let $\eta(t)$ be a path in $\Diff(M)$ through the identity such that $\dot\eta(0) = u \in \Xcal(M)$.
Then
\begin{equation}
	\ad_u(v) = \frac{\ud}{\ud t}\Bigg|_{t=0} \Ad_{\eta(t)}v = \frac{\ud}{\ud t}\Bigg|_{t=0} D\eta(t)(v\circ \eta(t)^{-1})
	= (Du)v + (Dv)\frac{\ud}{\ud t}\Bigg|_{t=0}\eta(t)^{-1}
\end{equation}
Just as in the finite dimensional case, we can obtain the time derivative of $\eta(t)^{-1}$ by
\begin{equation}\label{eq:diffinv_derivative}
	0 = \frac{\ud}{\ud t}\left(\eta(t)^{-1}\circ\eta(t)\right) = \left(\frac{\ud}{\ud t}\eta(t)^{-1}\right)\circ\eta(t) + \left(D\eta(t)^{-1}\circ\eta(t)\right)\dot\eta(t).
\end{equation}
Taking $t=0$ we thereby obtain
\begin{equation}
	\frac{\ud}{\ud t}\Bigg|_{t=0}\eta(t)^{-1} = -u.
\end{equation}
Thus,
\begin{equation}\label{eq:bracket_vf}
	\ad_u(v) = (Du)v - (Dv)u \eqqcolon [u,v]
\end{equation}
which is minus the standard commutator of vector fields.

We are now going to describe how the incompressible Euler equation arise as Euler--Poincaré equations.
To this end, we first need to identify a suitable infinite dimensional Lie group.
Recall from \autoref{sec:incompressible_euler} that the time dependent vector field $v$ describing the motion of the fluid should be divergence free.
This means that $\Diff(M)$ is too large, since its algebra contains all (tangential) vector field.
Instead, let $M$ be Riemannian and consider the subgroup of \emph{volume preserving diffeomorphisms}
\begin{equation}
	\SDiff(M) = \{\varphi\in\Diff(M)\mid \varphi_*\vol = \vol \}
\end{equation}
where $\vol$ is the Riemannian volume form.
Equivalently, a diffeomorphism $\varphi$ is volume preserving if $\det(D\varphi) \equiv 1$.
It is clear that $\SDiff(M)$ is a subgroup, because $\varphi_*\vol = \eta_*\vol = \vol$ implies that
\begin{equation}
	(\varphi\circ\eta)_*\vol = \varphi_*(\eta_*\vol) = \varphi_*\vol = \vol.
\end{equation}

The next step is to compute the subalgebra of $\Xcal(M)$ corresponding to $\SDiff(M)$.
If $\varphi(t)$ is a path in $\SDiff(M)$ through the identity with $\dot\varphi(0) = v$, then
\begin{equation}
	\frac{\ud}{\ud t}\Bigg|_{t=0}\varphi(t)_*\vol = 0 \iff \LieD_v \vol = 0 \iff \divv v = 0.
\end{equation}
Thus, the algebra of $\SDiff(M)$ consists of the divergence free vector fields
\begin{equation}
	T_{\id}\SDiff(M) = \Xcalvol(M) \coloneqq  \{ v\in \Xcal(M)\mid \divv v = 0 \}.
\end{equation}

Coming back to Euler, the intuition behind $\SDiff(M)$ as a configuration space is that the motion of the fluid particles is described by a path $t\mapsto \varphi(\cdot,t)$ of volume preserving diffeomorphisms: the position at time $t$ of the fluid particle initially at the point $x\in M$ is then $\varphi(x,t)$.
Thus, the kinetic energy is the integral over $M$ of all infinitesimal fluid particles: this gives us the Lagrangian
\begin{equation}\label{eq:lagrangian_non_invariantL2}
	L(\varphi,\dot\varphi) = \frac{1}{2}\int_M \abs{\dot\varphi(x)}^2\vol 	
\end{equation} 
where we have assumed that the mass density is $1$.
Notice that $L$ is quadratic, non-degenerate positive form in the variable $\dot\varphi$.
Hence, it defines an $L^2$-type Riemannian structure on $\SDiff(M)$ (or more generally on $\Diff(M)$):
\begin{equation}\label{eq:semi_invariant_L2_metric}
	\pair{U,V}_{L^2} = \int_M U(x)\cdot V(x)\,\vol ,\qquad U,V\in T_\varphi\SDiff(M).
\end{equation}

\begin{lemma}\label{lem:right_invariant_L2}
	The Riemannian metric \eqref{eq:semi_invariant_L2_metric} on $\SDiff(M)$ (and hence also the Lagrangian function \eqref{eq:lagrangian_non_invariantL2} on $T\SDiff(M)$) is right invariant.
	That is, for any $\eta\in\SDiff(M)$ and any $U,V\in T_\varphi\SDiff(M)$ we have
	\begin{equation}
		\pair{U\circ\eta,V\circ\eta}_{L^2} = \pair{U,V}_{L^2}.
	\end{equation}
\end{lemma}

\begin{proof}
	We have
	\begin{equation}
		\pair{U\circ\eta,V\circ\eta}_{L^2} = \int_M (U\cdot V)\circ\eta \,\vol .
	\end{equation}
	By $\eta^*\vol = \vol$ and the integral change of variables formula we also have, for any function $f$ on $M$, that
	\begin{equation}
		\int_M f\circ\eta \,\vol = \int_M f\circ\eta \,\eta^*\vol = \int_{\eta(M)}f\,\vol = \int_{M}f\,\vol.
	\end{equation}
	The result now follows by taking $f = U\cdot V$.
\end{proof}

\begin{remark}
	Notice that the Riemannian metric \eqref{eq:semi_invariant_L2_metric} extended to all of $\Diff(M)$ is \emph{not} right invariant with respect to $\Diff(M)$, but only with respect to the subgroup $\SDiff(M)$.
	This is important, because it implies that for compressible Euler equations, where the configuration space is $\Diff(M)$ with kinetic energy still corresponding to the $L^2$ metric \eqref{eq:semi_invariant_L2_metric}, it is not possible to reduce the equation to the Lie algebra $\Xcal(M)$; in addition to the fluid vector field one needs the time dependent mass density as a state variable.\todo{Add reference to coming section on compressible fluids.}
\end{remark}

Recall that a curve $[0,1]\ni t \to \varphi(t) \in \SDiff(M)$ is a geodesic with respect to the Riemannian metric~ if it extremizes the length functional
\begin{equation}
	\mathrm{len}(\varphi) = \int_0^1 \sqrt{\pair{\dot\varphi(t),\dot\varphi(t)}_{L^2}}\,\ud t.
\end{equation}
The length functional is independent of the parameterization of $\varphi$.
Thus, we may choose the constant speed parameterization $\pair{\dot\varphi,\dot\varphi}_{L^2} = \text{const}$.

Our objective is to prove the following remarkable result, which can be considered the foundation of the field of geometric hydrodynamics.

\begin{theorem}[\citet{Ar1966}]\label{thm:arnold1966}
	Let $\varphi\colon[0,1]\to \SDiff(M)$ be a constant speed geodesic with respect to the Riemannian metric \eqref{eq:semi_invariant_L2_metric}.
	Then the time dependent vector field
	\begin{equation}
		v(t) = \dot\varphi(t)\circ\varphi(t)^{-1}
	\end{equation}
	fulfills the incompressible Euler equation
	\begin{equation}
		\dot v + \nabla_v v = -\nabla p, \qquad \divv v = 0.
	\end{equation}
\end{theorem}

Before we go on to the proof, let us give some examples on how the result can be used for insights on fluid motion:
\begin{enumerate}
	\item The stability of the fluid motion is related to the sectional curvature of the Riemannian metric \eqref{eq:semi_invariant_L2_metric} on $\SDiff(M)$: in regions of positive curvature leads to converging motion of the fluid particles (laminar type flows), whereas regions with negative curvature leads to diverging fluid motion (turbulent flows).
	Arnold showed that the sectional curvature is negative in almost all directions, so that nearby fluid regions typically diverge exponentially fast.
	This led him to the result that reliable long-term weather forecast are practically impossible.

	\item Arnold's framework for fluids can be used to give rigorous local well-posedness results for the Euler equations.
	This was done in 1970 by \citet{EbMa1970}, who showed that, when formulated using Lagrangian coordinates on $T\SDiff(M)$, the resulting dynamical has no loss of derivatives, so it can be extended to a smooth ordinary differential equation on a Banach manifold of Sobolev type.
	Local existence and uniqueness, as well as smooth dependence on initial conditions, then follows from standard ODE techniques (Picard--Lindelöf iterations).

	\item In addition to the incompressible Euler equation, many PDE of mathematical physics has been realized to fit the framework of Arnold, but for different infinite dimensional groups, and different choices of Riemannian metrics.
	For example, the KdV, Camassa--Holm, Landau--Lifschitz, and magneto hydrodynamical equations are all examples.
	Today, such equations are called \emph{Euler--Arnold equations}.
\end{enumerate}

In order to prove \autoref{thm:arnold1966}, let us first recollect the close relation between Lagrangian mechanics and Riemannian geometry:
\begin{framed}
	\centering
	$\varphi$ extremizes the action functional for the kinetic energy Lagrangian \eqref{eq:lagrangian_non_invariantL2}.
	\\[1ex]
	$\Updownarrow$
	\\[1ex]
	$\varphi$ is a constant speed geodesic curve for the Riemannian metric \eqref{eq:semi_invariant_L2_metric}.
\end{framed}
Thus, analogous to the finite dimensional setting of \autoref{sec:mechanics_on_lie_groups}, we shall study the Euler--Lagrange equations for the Lagrangian \eqref{eq:lagrangian_non_invariantL2} on $T\SDiff(M)$.
First we need a result analogous to \autoref{lem:variation_poincare}.

\begin{lemma}\label{lem:variation_arnold}
	Let $\varphi\colon [0,1]\times [-\delta,\delta]\to \SDiff(M)$ and consider the right reduced tangent paths on $\Xcalvol(M)$ given by
	\begin{equation}
		v(t,\epsilon) \coloneqq \frac{\pd \varphi(t,\epsilon)}{\pd t}\circ\varphi(t,\epsilon)^{-1} 
	\end{equation}
	and
	\begin{equation}
		u(t,\epsilon) \coloneqq \frac{\pd \varphi(t,\epsilon)}{\pd \epsilon}\circ\varphi(t,\epsilon)^{-1}.
	\end{equation}
	Then
	\begin{equation}
		\frac{\pd v}{\pd \epsilon} - \frac{\pd u}{\pd t} = -[v,u].
	\end{equation}
\end{lemma}

\begin{proof}
	Direct calculations give
	\begin{equation}
		\frac{\pd v}{\pd \epsilon} = \frac{\pd^2 \varphi}{\pd t\pd \epsilon}\circ\varphi^{-1} + \left(D\frac{\pd \varphi}{\pd t}\circ\varphi^{-1}\right)\frac{\pd \varphi^{-1}}{\pd\epsilon}
	\end{equation}
	and
	\begin{equation}
		\frac{\pd u}{\pd t} = \frac{\pd^2 \varphi}{\pd t\pd \epsilon}\circ\varphi^{-1} + \left(D\frac{\pd \varphi}{\pd \epsilon}\circ\varphi^{-1}\right)\frac{\pd \varphi^{-1}}{\pd t}
	\end{equation}
	From the calculation \eqref{eq:diffinv_derivative} we see that
	\begin{equation}
		\frac{\pd \varphi^{-1}}{\pd t} = -D\varphi^{-1}\left( \frac{\pd \varphi}{\pd t}\circ\varphi^{-1} \right) = -D\varphi^{-1}v
	\end{equation}
	and
	\begin{equation}
		\frac{\pd \varphi^{-1}}{\pd \epsilon} = -D\varphi^{-1}\left( \frac{\pd \varphi}{\pd \epsilon}\circ\varphi^{-1} \right)=-D\varphi^{-1}u.
	\end{equation}
	This gives
	\begin{equation}
		\frac{\pd v}{\pd \epsilon} - \frac{\pd u}{\pd t} = -\left(D\frac{\pd \varphi}{\pd t}\circ\varphi^{-1}\right)D\varphi^{-1}u + \left(D\frac{\pd \varphi}{\pd \epsilon}\circ\varphi^{-1}\right)D\varphi^{-1}u
	\end{equation}
	From the chain rule we have
	\begin{equation}
		\left(D\frac{\pd \varphi}{\pd t}\circ\varphi^{-1}\right)D\varphi^{-1}u =  D(\underbrace{\frac{\pd \varphi}{\pd t}\circ\varphi^{-1}}_{v}) u
	\end{equation}
	and similarly for the other term.
	Thus,
	\begin{equation}
		\frac{\pd v}{\pd \epsilon} - \frac{\pd u}{\pd t} = -(Dv)u + (Du)v = -[v,u].
	\end{equation}
	This proves the theorem.
\end{proof}

We are now in position to prove Arnold's theorem.

\begin{proof}[Proof of \autoref{thm:arnold1966}]
	Let $\varphi(t)$ be a curve in $\SDiff(M)$ that extremizes the kinetic energy Lagrangian \eqref{eq:lagrangian_non_invariantL2}.
	Thus, if $\varphi_\epsilon$ is a variation of $\varphi(t)$, then, from right invariance it follows that
	\begin{equation}
		\frac{\ud}{\ud\epsilon}\Bigg|_{\epsilon=0}\frac{1}{2}\int_0^1 \int_M \abs{v_\epsilon(x,t)}^2 \vol \ud t = 0
	\end{equation}
	where $v_\epsilon = \dot\varphi\circ\varphi^{-1}$.
	This gives
	\begin{equation}
		\int_0^1 \int_M \pair{v(x,t), \frac{\ud}{\ud\epsilon}\Bigg|_{\epsilon=0} v_\epsilon(x,t)} \vol \ud t = 0.
	\end{equation}
	From \autoref{lem:variation_arnold} we then get
	\begin{equation}
		\int_0^1 \int_M \pair{v, \dot u - [v,u]} \vol \ud t = 0,
	\end{equation}
	where $u\colon[0,1]\to \Xcalvol(M)$ is an arbitrary path which vanishes at the end points.
	Using that the Levi-Civita connection fulfills
	\begin{equation}
		\nabla_v u - \nabla_u v = [u,v]
	\end{equation}
	we get
	\begin{equation}
		0 = \int_0^1 \int_M \pair{v, \dot u + \nabla_v u - \nabla_u v} \vol \ud t = \int_0^1 \int_M \left(\pair{v, \dot u + \nabla_v u} - \frac{1}{2}\pair{\nabla\abs{v}^2,u}\right) \vol \ud t.
	\end{equation}
	Since $u$ is divergence free, it follows from Helmholtz \autoref{lem:helmholtz} that the last term vanishes.
	Integration by parts in time, using that $u$ vanishes at the end points, and the co-variant derivative identity
	\begin{equation}
		\pair{w,\nabla (v\cdot u)} = \pair{\nabla_w v,u} + \pair{v,\nabla_w u}
	\end{equation}
	then leads to
	\begin{equation}
		\int_0^1 \left(\pair{-\dot v - \nabla_v v , u}_{L^2} + \pair{v,\nabla (v\cdot u)}_{L^2}\right) \ud t = 0
	\end{equation}
	Again, using Helmholtz \autoref{lem:helmholtz} we finally obtain
	\begin{equation}
		\int_0^1 \pair{\dot v + \nabla_v v , u}_{L^2} \ud t = 0 .
	\end{equation}
	This condition should be valid for any path $u(t)$ in $\Xcalvol(M)$.
	However, the expression $\nabla_v v$ is in general not divergence free, so we cannot conclude that $\dot v + \nabla_v v = 0$.
	Rather, we can conclude that $\dot v + \nabla_v v$ must be $L^2$ orthogonal in $\Xcal(M)$ to the subspace $\Xcalvol(M)$.
	From the Helmholtz decomposition this implies that
	\begin{equation}
		\dot v + \nabla_v v + \nabla p = 0
	\end{equation}
	for some function $p$ unique up to addition of constants.
\end{proof}

\iftoggle{coimbra2}{
	The aim of the lectures notes is now fulfilled: to present and prove Arnold's \autoref{thm:arnold1966}.
	Let me point out, however, that this is only where the story of geometric hydrodynamics begins.
	Since Arnold's work, other mathematicians and physicists have found many other PDE in mathematical possessing the same structure.
	The active area of research has also shown connections to optimal transport theory, shape analysis, information theory, Kähler geometry, and other fields of mathematics. 
	I refer to the monographs by \citet{ArKh1998} and \citet{KhWe2009} for a deeper study.
	To encourage further reading, I list here some interesting research directions:
	\begin{itemize}
		\item Shallow water equations are essential for our understanding of oceanic and atmospheric flows. 
		It turns out that many of those equations also has can be interpreted as flows on diffeomorphism groups, but now with respect to a Riemannian metric that is not fully right-invariant, and also with an added potential function (typically originating from gravity).
		To study shallow water equations from the point-of-view of Arnold yields new insights into higher-order conservation laws, existence and uniqueness properties, as well as construction of efficient and structure preserving numerical method. \cite{Ha1982,Io2012,BaMo2018,MoVi2018}

		\item Optimal mass transport is an old mathematical problem initially formulated by the French mathematician Gaspard Monge in the 18th century.
		During the last 20 years this field has gone through enormous mathematical developments.\footnote{Modern mathematical work in optimal transport theory has (so far) generated two Fields medals; Cédric Villani (2010) and Alessio Figalli (2018).}\
		Many of these developments have to do with the combination of (infinite-dimensional) geometry and analysis.
		Directly connected to hydrodynamics and Euler--Arnold equations is the work by \citet{BeBr2000} and \citet{Ot2001} who realized that the $L^2$-Wasserstein distance in optimal transport can be interpreted as coming from an infinite-dimensional Riemannian metric on the space of probability densities. In turn, this Riemannian metric is nothing but Arnold's metric on the space of diffeomorphisms, but restricted to so-called \emph{horizontal directions}, as is well-explained in Appendix A.5 of the monograph by \citet{KhWe2009}.
		In addition to geodesic equations one may also study Riemannian gradient flows.
		This line of research was initiated by \citet*{JoKiOt1998}, who realized that the heat equation in physics can be interpreted as the Riemannian gradient flow of the entropy functional, thereby `proving' the second law of thermodynamics that the entropy of an isolated system increases.

		\item Geometric hydrodynamics offers an interesting way to view the Schrödinger equation in quantum mechanics.
		Indeed, by extending Arnold's Riemannian metric to the space of all diffeomorphisms (as in the connection to optimal transport), and then adding a potential function given by the \emph{Fisher information functional} (well known in statistics), one can, via a \emph{Madelung transform} obtain a link between Schrödinger equations and hydrodynamics: the Schrödinger equation becomes a compressible fluid equation with a non-Newtonian potential.
		It turns out that the Madelung transform has many interesting geometric properties.
		In particular, it is a Kähler mapping between the complex projective Hilbert space and the co-tangent bundle of the space of probability densities equipped with the (lifted) Fisher--Rao metric (also well known in statistics). \cite{Ma1927,KhMiMo2018,KhMiMo2019}
	\end{itemize}
}{
\section{Hamiltonian Formulation and Vorticity}

\section{Other Examples}
}





\def\cprime{$'$}

\end{document}